\documentclass[12pt]{article}

\usepackage[utf8]{inputenc}
\usepackage{geometry}
\usepackage{setspace}
\usepackage[dvipdfmx]{graphicx}
\usepackage{color}
\usepackage{booktabs} 
\usepackage{array} 
\usepackage{paralist} 
\usepackage{latexsym}
\usepackage{amssymb}
\usepackage{amsfonts}
\usepackage{bm}
\usepackage{natbib}
\usepackage{comment}
\usepackage{ascmac}
\usepackage{epsfig}
\usepackage{epsf}
\usepackage{float}
\usepackage{epstopdf}
\usepackage{url}
\usepackage{authblk}
\usepackage{textcomp}
\usepackage{algorithm}
\usepackage{algpseudocode}
\usepackage{setspace}

\usepackage {yhmath}
\usepackage{amsmath}

\usepackage{amsthm}
\theoremstyle{plain}
\newtheorem{thm}{Theorem}[section]
\newtheorem{remark}[thm]{Remark}

\title{Parallelizing MCMC with Machine Learning Classifier and Its Criterion Based on Kullback–Leibler Divergence}
\author{
Tomoki Matsumoto \thanks{Corresponding author: T.Matsumoto, Dr. of Engineering,
   Assistant Professor (Specially Appointed), School of Economics, University of Toyama,
    3190 Gofuku, Toyama City, Toyama, 930-8555, Japan,
    E-mail: t.matsumoto514@gmail.com, mtomoki@eco.u-toyama.ac.jp, 
    16-digit ORCID ID: 0000-0002-5680-7431} 
}

\begin{document}
\bibliographystyle{apalike}

\maketitle

\newpage
\begin{abstract}
In the era of Big Data, Markov chain Monte Carlo (MCMC) methods, which are currently essential for Bayesian estimation, face significant computational challenges owing to their sequential nature.
To achieve a faster and more effective parallel computation, we emphasize the critical role of the overlapped area of the posterior distributions based on partitioned data, which we term the reconstructable area. 
We propose a method that utilizes machine learning classifiers to effectively identify and extract MCMC draws obtained by parallel computations from the area based on posteriors based on partitioned sub-datasets, approximating the target posterior distribution based on the full dataset. 
This study also develops a Kullback-Leibler (KL) divergence-based criterion. 
It does not require calculating the full-posterior density and can be calculated using only information from the sub-posterior densities, which are generally obtained after implementing MCMC. 
This simplifies the hyperparameter tuning in training classifiers. 
The simulation studies validated the efficacy of the proposed method. 
This approach contributes to ongoing research on parallelizing MCMC methods and may offer insights for future developments in Bayesian computation for large-scale data analyses.


\vspace{3mm}
\noindent \textbf{keywords:} {Big data; Bayesian computation; Parallelizing MCMC; Machine learning classifier}
\end{abstract}

\newpage

\section{Introduction}

Advancements in information technology have enabled the collection of high-dimensional and voluminous {\lq\lq Big Data\rq\rq} across various fields, such as finance, healthcare, and social media. 
Bayesian statistics, which is known for handling uncertainty and flexible model building, is suitable for these complex data structures. 
However, Markov chain Monte Carlo (MCMC) methods, often used in Bayesian estimation and crucial for generating outcomes or samples from a distribution or a random variable, face significant challenges.
Owing to their inherently sequential nature, MCMCs often result in prolonged computational time when applied to large-volume and high-dimensional datasets.
Several strategies for a faster MCMC have been proposed to address this problem through parallel computing. 
However, naive parallelization algorithms, including stochastic gradient MCMC methods, incur high communication costs between machines and undermine efficiency \citep{Scott2013, Nishihara2014, Dunson2023}.

To overcome these machine-communication issues, some methods have demonstrated effective parallelization with minimal communication under the independent product equation (IPE) condition, as detailed in Section \ref{IPE}.
Their approach allows a dataset to be divided into sub-datasets and processed independently across different processors or machines. 
They implemented MCMC independently in each processor, without communication, until the final step, in which multiple results are combined.
This facilitates {\lq\lq one-shot learning\rq\rq} as referred to by \cite{Dunson2023}. 
Thus, the main interest is how multiple results should be combined to approximate draws from the posterior distribution based on the full dataset to achieve one-shot parallelizing MCMC.

Several methods have been developed under IPE to achieve one-shot parallelizing MCMC. 
\cite{Scott2013} introduced the Consensus Monte Carlo (CMC), combining sub-posterior draws by weighted averaging, assuming asymptotic posterior normality based on the Bernstein–von Mises theorem (\cite{Van1998}). 
\cite{Neiswanger2014} proposed an embarrassingly parallel MCMC method using kernel density estimation for posterior approximation. 
\cite{Dunson2014} introduced the Weierstrass rejection sampler, employing the Weierstrass transform to approximate sub-posterior distributions and the Gibbs sampler to combine multiple results for sampling from the approximated full-posterior distribution. 
The latter two approaches require careful bandwidth selection of the kernel. 
Although \cite{Nemeth2018} generalized CMC by modeling the log sub-posterior density with a Gaussian process, their method necessitates ad hoc adjustments for the mean function and kernel for the covariance matrix.
These methods generally suffer from hyperparameter tuning of the algorithms.

In this study, we propose a parallelizing MCMC method using classification models from the machine learning field. 
We emphasize the importance of the overlapped area of sub-posteriors, termed the reconstructable area in this study, for combining the results from each MCMC.
To specify the reconstructable area, our method utilizes machine learning classifiers, especially random forest (\cite{Breiman2001}), which is known for their effectiveness in handling high-dimensional data, to combine MCMC draws from posterior distributions based on sub-datasets\footnote{In the following, we refer this as sub-posterior draws.} to offer an alternative approach for parallelizing MCMC within a one-shot learning framework. 
In general, machine learning classifiers require hyperparameter tuning to improve their performance.
This paper also proposes a criterion based on the Kullback-Leibler (KL) divergence to select machine learning classifiers for the proposed parallelizing MCMC method.
The criterion can be calculated using only information on the posterior distributions after data partitioning and helps reduce the effort required for hyperparameter tuning.
Two simulations were conducted to validate the proposed method. The first simulation in a multivariate posterior case indicated that our KL divergence-based criterion correlated well with the value of KL divergence between the true posterior and approximated posterior distributions.
The second simulation also suggests that the proposed method can combine sub-posteriors well, even if the posterior distribution has multiple modes.

The remainder of this paper is organized as follows. 
Section \ref{IPE} explains IPE condition, which is essential for parallelizing MCMC methods with minimal communication. 
Section \ref{parallelizing MCMC} describes the proposed parallelizing MCMC method for combining posterior draws with machine learning classifiers and the KL divergence-based criterion for classifier selection. 
Section \ref{simulation} presents simulation studies that validate the proposed method. 
Finally, Section \ref{conclusion} concludes the paper with a discussion of the remaining issues of our method.

\section{Fundamental idea of parallelizing MCMC}\label{IPE}

For the later discussion, we explain what parallelizing MCMC means in terms of avoiding high communication costs.
By parallelizing MCMC, communication between the master node and each processor or machine can be kept minimal, and almost all procedures can be independently performed within each processor.
Figure \ref{fig:concept_of_parallelizingMCMC} shows the basic concept of minimum communication, that is, one-shot parallelizing MCMC.
In this approach, data is first divided into $m$ independent sub-datasets, MCMC draws are obtained from each sub-posterior distribution, and finally, the multiple results are combined.
Figure \ref{fig:concept_of_parallelizingMCMC} shows this concept.

\begin{figure}
    \centering
    \includegraphics[width=0.6\linewidth]{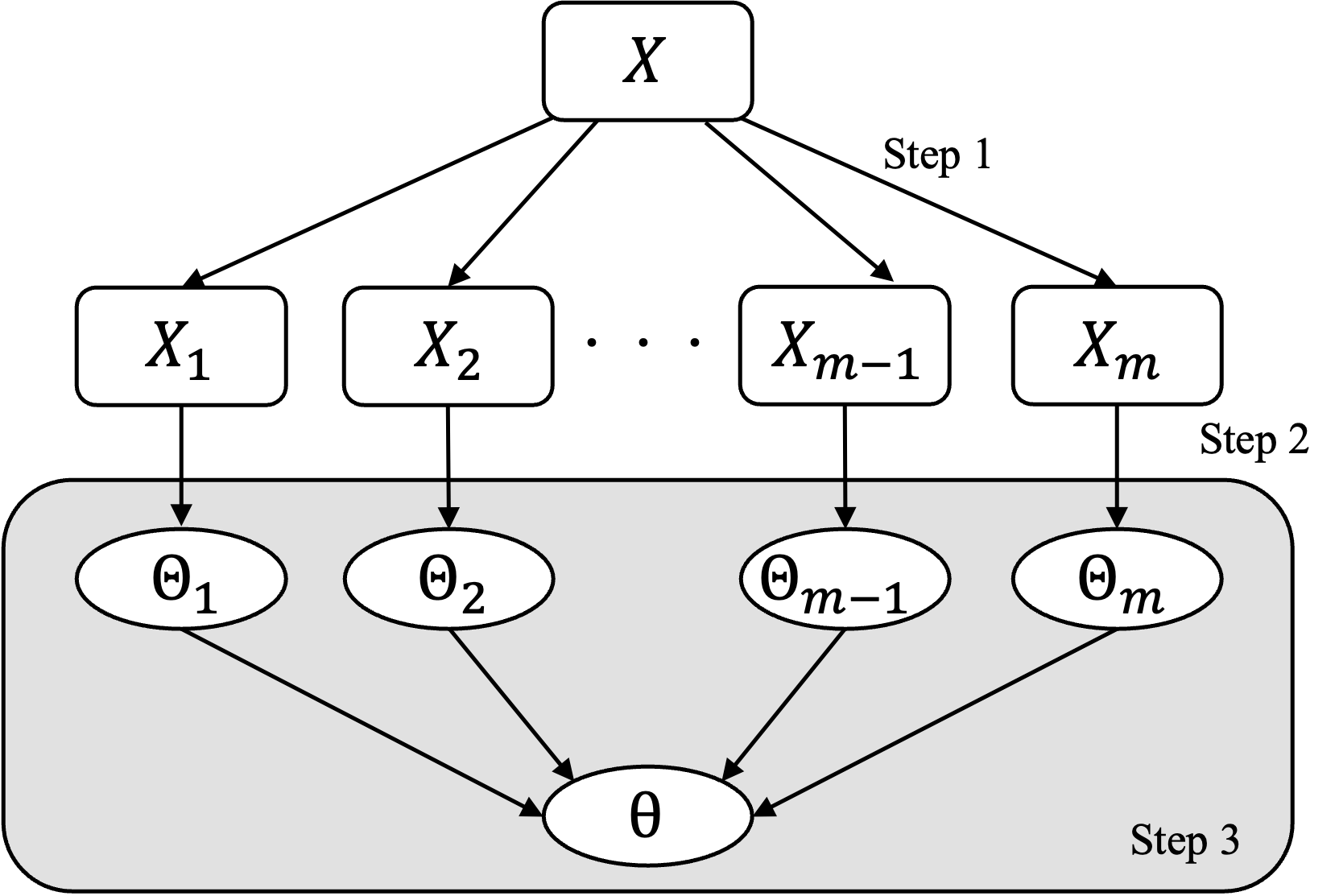}
    \caption{Concept of one-shot parallelizing MCMC. In Step 1, full-dataset is divided into $m$ mutually independent sub-datasets. In Step 2, sub-posterior draws are sampled in each processor. In Step 3, the multiple results are combined to approximate posterior draws based on full-dataset.}
    \label{fig:concept_of_parallelizingMCMC}
\end{figure}

We define new notations for a parametric model $p(X|\theta)$.
Let $X$ be $n \times b$ data matrix or full-dataset and $\theta$ be $d$-dimensional parameter vector for the parametric model.
The posterior distribution $\pi(\theta|X)$ for the full-dataset $X$ with a prior distribution $\pi(\theta)$ can be expressed by Bayes' theorem as follows:
\begin{align}
\pi(\theta|X)= \frac{p(X|\theta)\pi(\theta)}{\int p(X|\theta)\pi(\theta) d\theta} \propto p(X|\theta)\pi(\theta). \label{eq:bayes_theorem}
\end{align}
The analyst chose the statistical model $p(X|\theta)$ and prior distribution $\pi(\theta)$ to analyze the full-dataset $X$.

To implement parallelizing MCMC, we divide $X$ into $m$ row-wise sub-datasets $X_i, \ i=1,\ldots,m$, where $X_i$ is $n_i \times b$ sub-data matrix ($\sum_{i=1}^{m}n_i=n$) and they are distributed to processors or machines.
We assume $m$ sub-datasets $X_{i}$, $i=1,\ldots,m,$ are independent conditional on parameter $\theta$, that is, $p(X|\theta)=\prod_{i=1}^{m}p(X_{i}|\theta)$.\footnote{This assumption can be satisfied, for example, by randomly sampling rows from $X$ without replacement to create each $X_{i}$.} 
For the model based on the sub-dataset $X_{i}$, $i=1,\ldots,m$, we define $\pi_i(\theta)$ as the prior distribution and $\pi(\theta|X_i) \propto p(X_{i}|\theta)\pi_{i}(\theta)$ as the posterior distribution.
The posterior distribution (\ref{eq:bayes_theorem}) can be calculated as follows:

\begin{align}
\pi(\theta|X) &\propto p(X|\theta)\pi(\theta) = \left\{ \prod_{i=1}^{m}p(X_i|\theta)\right\} \pi(\theta) = \left\{\prod_{i=1}^{m}\frac{p(X_{i}|\theta)\pi_{i}(\theta)}{\pi_i(\theta)}\right\} \pi(\theta) \nonumber \\
&\propto \left\{\prod_{i=1}^{m}\pi(\theta|X_i)\right\}\left\{\frac{\pi(\theta)}{\prod_{i=1}^{m}\pi_{i}(\theta)}\right\}.
\label{eq:posteriordistribution}
\end{align}

\noindent Furthermore, when we require $\pi(\theta) \propto \prod_{i=1}^{m}\pi_{i}(\theta)$, equation $(\ref{eq:posteriordistribution})$ above becomes
\begin{equation}
\pi(\theta|X)\propto \prod_{i=1}^{m}\pi(\theta|X_i).
\label{eq:independentproductequation}
\end{equation}
To satisfy the requirement for prior distribution and preserve the total amount of prior information, a simple method is to set $\pi_{i}(\theta) \propto \pi(\theta)^{1/m}$.
Note that the prior setting for each sub-dataset is solely for implementing a parallelizing MCMC method, and the true prior chosen by the analyst remains $\pi(\theta)$

Equation (\ref{eq:independentproductequation}) means that under the independence assumption of the data partition and a little prior requirement, the posterior distribution $\pi(\theta|X)$ based on the full-dataset can be represented by the
product of sub-posterior distributions $\pi(\theta|X_{i})$, and this should make the computation of each MCMC much faster. 
\citet{Dunson2014} calls equation $(\ref{eq:independentproductequation})$ above the independent product equation (IPE).
In this study, parallelizing MCMC is the name given to the collection of several methods combining posterior draws based on the sub-datasets to manufacture draws to approximate draws based on the full-dataset, taking advantage of structures under IPE.
After constructing sub-posterior distributions under IPE, we sample MCMC draws from each sub-posterior and combine them to manufacture or approximate draws from a full-posterior distribution.
The following section proposes and explains how to combine sub-posterior draws to achieve parallelizing MCMC.

\section{Machine learning classifier-based method and Kullback-Leiber-based criterion} \label{parallelizing MCMC}
This section proposes a method that combines multiple sets of MCMC draws from sub-posterior distributions using a machine learning classifier. 
We explain that the area where multiple MCMC draws from $\pi(\theta|X_{i})$, $i=1\ldots,m$, overlap is the crucial support required to reconstruct the full-posterior. 
We then demonstrate how to identify such support.

\subsection{Strategy of combining posterior draws} \label{Strategy of combining posterior draws}
The Weierstrass transform $W_{h}$ is a convolution of a normal distribution with mean $0$ and bandwidth $h$ for a probability distribution:
\begin{align*}
    W_{h}g(\theta) = \int_{-\infty}^{\infty} \frac{1}{\sqrt{2\pi}h}\exp\left\{-\frac{(\theta-y)^{2}}{2h^{2}}\right\}g(y)dy.
\end{align*}
As $h \rightarrow 0$, the Weierstrass transform recovers the original distribution because the normal random variable added to the Weierstrass transform converges to its mean of zero.
By reinterpreting \cite{Dunson2014}, this subsection shows that the overlapped area of the sub-posterior distributions plays an essential role in our parallelizing MCMC method by exploring how the full-posterior distribution is recovered when the Weierstrass transform is adopted for the sub-posterior distributions and bandwidth parameters move to 0 under IPE.

Under IPE, $d$-dimensional full-posterior can be approximated via the (multivariate) Weierstrass transform $W_{h}$ with a common bandwidth $h$ used for all dimensions as follows:

\begin{eqnarray}
\pi(\theta|X) &=& \prod_{i=1}^{m}\pi(\theta|X_{i}) \approx \prod_{i=1}^{m}W_{h}\pi(\theta|X_{i}) \nonumber\\
&=&\prod_{i=1}^{m}\int\frac{1}{(2\pi)^{\frac{d}{2}}|h^{2}I|^{\frac{1}{2}}}\exp\left\{-\frac{(\theta-{r}_{i})^{\top}(\theta-{r}_{i})^{\top}}{2h^{2}}\right\}\pi({r}_{i}|X_{i})d{r_{i}}\nonumber\\
&\propto&\int\exp\left\{-\frac{\sum_{i=1}^{m}(\theta-{r}_{i})^{\top}(\theta-{r}_{i})}{2h^{2}}\right\}\pi({r}_{1}|X_{1})\cdots \pi({r}_{m}|X_{m})d{r},
\label{eq:ipeddimensionalcase}
\end{eqnarray}

\noindent where the last integral represents a multiple integral with respect to $r_{1}, \ldots, r_{m}$ and $I$ is $d \times d$ identity matrix.
For the inside of the braces in (\ref{eq:ipeddimensionalcase}), we can further calculate as follows:
\begin{eqnarray*}
\sum_{i=1}^{m}(\theta-{r}_{i})^{\top}(\theta-{r}_{i})&=&\sum_{i=1}^{m}(\theta^{\top}\theta-\theta^{\top}{r}_{i}-{r}_{i}^{\top}\theta+{r}_{i}^{\top}{r}_{i})\\
&=&m\left(\theta^{\top}\theta-\theta^{\top}\cdot \frac{1}{m}\sum_{i=1}^{m}{r}_{i}-\left(\frac{1}{m}\sum_{i=1}^{m}{r}_{i}^{\top}\right)\theta+\frac{1}{m}\sum_{i=1}^{m}{r}_{i}^{\top}{r}_{i}\right).
\end{eqnarray*}

\noindent Let $\bar{R}=1/m\sum_{i=1}^{m}{r}_{i}$, and $\bar{R^{2}}=1/m\sum_{i=1}^{m}{r}_{i}^{\top}{r}_{i}$, then the expression above is rewritten as
\begin{eqnarray}
\sum_{i=1}^{m}(\theta-{r}_{i})^{\top}(\theta-{r}_{i})
&=&m\left(\theta-\bar{R}\right)^{\top}\left(\theta-\bar{R}\right)+m\left(\bar{R^{2}}-\bar{R}^{\top}\bar{R}\right).
\label{eq:insideofbraces1}
\end{eqnarray}

\noindent Hence, with $C_{1}$ and $C_{2}$ being normalizing constants, equation (\ref{eq:ipeddimensionalcase}) along with (\ref{eq:insideofbraces1}) can be rewritten as
\begin{eqnarray}
&&W_{h}\pi(\theta|X)\nonumber\\
&=&C_{1}^{-1}\int\exp\left\{-\frac{m\left(\theta-\bar{R}\right)^{\top}\left(\theta-\bar{R}\right)+m\left(\bar{R^{2}}-\bar{R}^{\top}\bar{R}\right)}{2h^{2}}\right\}\cdot \pi({r}_{1}|X_{1})\cdots \pi({r}_{m}|X_{m})d{r}\nonumber\\
&=&C_{1}^{-1}\int\exp\left\{-\frac{m\left(\theta-\bar{R}\right)^{\top}\left(\theta-\bar{R}\right)}{2h^{2}}\right\}\cdot \exp\left\{-\frac{m\left(\bar{R^{2}}-\bar{R}^{\top}\bar{R}\right)}{2h^{2}}\right\}\cdot \pi({r}_{1}|X_{1})\cdots \pi({r}_{m}|X_{m})d{r}\nonumber\\
&=&C_{2}^{-1}\int{\cal{N}}\left(\theta\left|\bar{R},\frac{h^{2}}{m}I\right.\right)\cdot \exp\left\{-\frac{m\left(\bar{R^{2}}-\bar{R}^{\top}\bar{R}\right)}{2h^{2}}\right\}\cdot \pi({r}_{1}|X_{1})\cdots \pi({r}_{m}|X_{m})d{r}.
\label{eq:approxedposteriorviamultiweierstrass}
\end{eqnarray}

\noindent The normalizing constant $C_{2}$ can be calculated as
\begin{eqnarray}
C_{2}
&=&E^{\pi(r|X)}\left[\exp\left\{-\frac{m\left(\bar{R^{2}}-\bar{R}^{\top}\bar{R}\right)}{2h^{2}}\right\}\right],
\label{eq:expectationofxpunderfullposterior}
\end{eqnarray}
where $E^{\pi(r|X)}[\cdot]$ expresses the expectation of an expression within a bracket under the full-posterior distribution under IPE.


As described above, $W_{h}\pi(\theta|X)$ converges to $\pi(\theta|X)$ pointwise as $h \to 0$ with the assumption of interchangeability between integration and taking a limit. 
However, in the braces within the exponential of (\ref{eq:expectationofxpunderfullposterior}), when $\bar{R^{2}}-\bar{R}^{\top}\bar{R}>0$, then $\exp\left\{-m(\bar{R^{2}}-\bar{R}^{\top}\bar{R})/(2h^{2})\right\} \to 0$ as $h \to 0$. 
Because the normalizing constant $C_{2}$ cannot be zero, we need to examine the area where $\bar{R^{2}}-\bar{R}^{\top}\bar{R}=0$. 
This implies that the support of integration with respect to ${r}$ in (\ref{eq:expectationofxpunderfullposterior}) can be reduced to ${\cal{E}}=\{({r}_{1},\ldots,{r}_{m})|\bar{R^{2}}-\bar{R}^{\top}\bar{R}=0\}$, or within area ${\cal{E}}$, the Weierstrass transform is valid. 
Thus, to the extent the Weierstrass transform is applicable and integration and taking limit can be interchanged, we only need to worry about the set ${\cal{E}}$ when we reconstruct the full-posterior from the posterior based on $X_{i}$, $i=1,\ldots,m$.\footnote{The same discussion applies to (\ref{eq:approxedposteriorviamultiweierstrass}).}
We call this area {\lq\lq reconstructable area,\rq\rq} illustrated in Figure \ref{fig:illustration_of_reconstructable_area}.

\begin{figure}[tbp]
 \begin{center}
  \includegraphics[width=120mm]{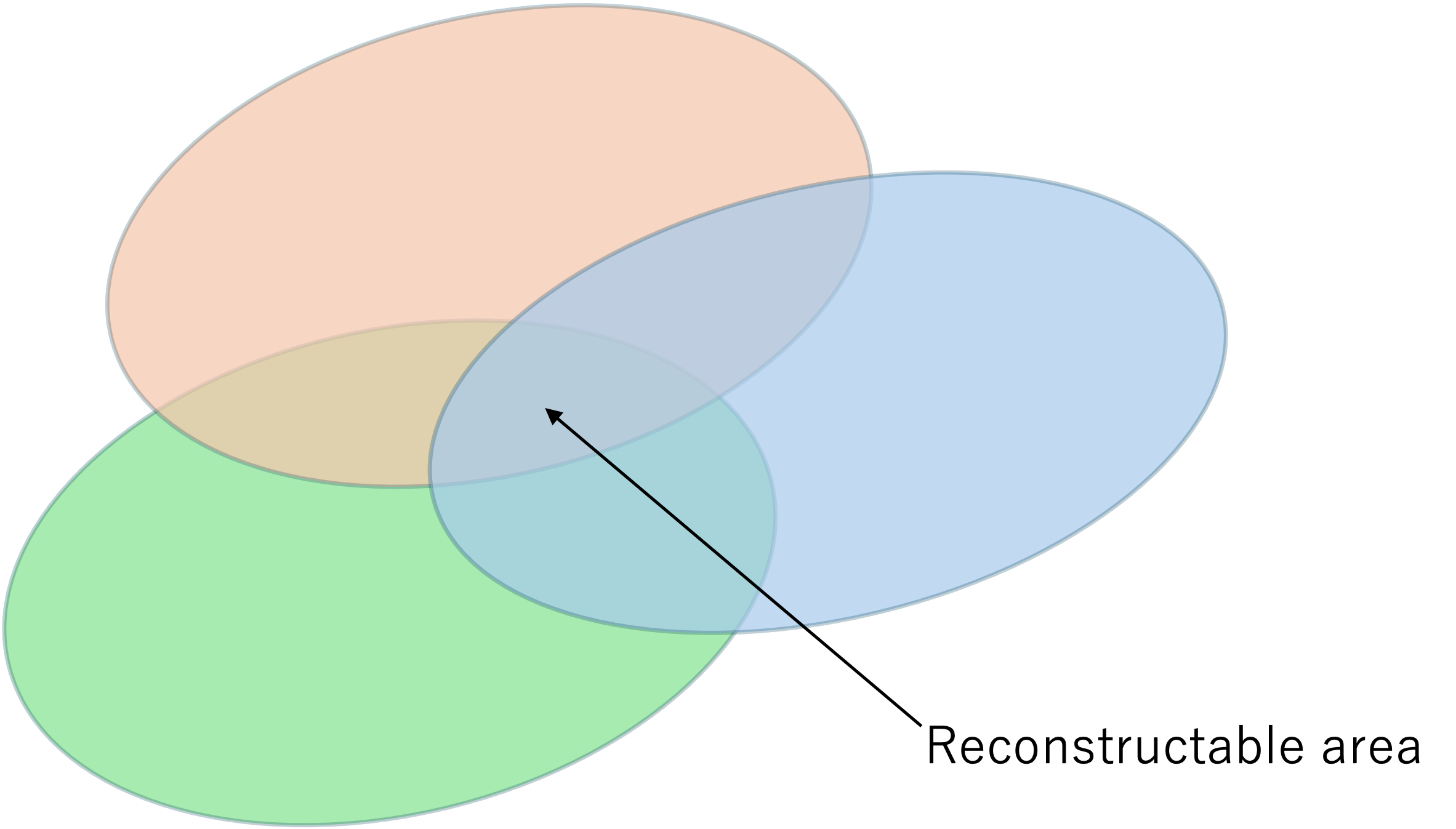}
 \end{center}
 \caption{Illustration of the reconstructable area in three data partition and two-dimensional posterior distributions.}
 \label{fig:illustration_of_reconstructable_area}
\end{figure}


Our task in parallelizing MCMC is to construct the distribution on the reconstructable area from posterior draws based on $X_{i}$.
Suppose that $N$ MCMC draws $\theta_{i}^{(t)},\ t=1,\ldots,N$, from $\pi(\theta|X_{i})$, $i=1,\ldots m$. Let
\begin{eqnarray*}
\Theta^{\top}&=&\left({\theta_{1}^{(1)}},\ldots,{\theta_{1}^{(N)}},\ldots,{\theta_{m}^{(1)}},\ldots,{\theta_{m}^{(N)}}\right)\\
&=&(\theta_{1}^{\ast},\ldots,\theta_{N}^{\ast},\ldots,\theta_{k}^{\ast},\ldots,\theta_{(m-1)N}^{\ast},\ldots,\theta_{mN}^{\ast}),
\end{eqnarray*}
hence $\Theta$ is $mN \times d$ matrix. 
For ease of notation, we re-index the MCMC draws using a single index $k$ that runs from $1$ to $mN$, where draws from the $i$-th sub-posterior correspond to $k = (i-1)N + 1,\ldots, iN$.
This allows us to reference any draw as $\theta_{k}^{*}$ regardless of which sub-posterior it came from.

Let $z \in \{1,\ldots,m\}$ be a variable expressing which processor a sub-posterior draw is sampled from.
We aim to determine whether $\theta_{k}^{\ast}$ belongs to the reconstructable area or not: if $\theta_{k}^{\ast}$ could have belonged equally to any one of $i=1,\ldots,m$ sub-processes, then we are able to say that $\theta_{k}^{\ast}$ belongs to the reconstructable area.
In other words, our basic idea for constructing a distribution on the reconstructable area is based on the following assumption: if we can calculate that $\theta_{k}^{\ast}$ could have been generated from any one of the $m$ sub-processes $\Pr(z=i|\theta_{k}^{\ast})$, $i=1,\ldots,m$, with equal probability $1/m$, then we regard $\theta_{k}^{\ast}$ with the draw from the reconstructable area.


\subsubsection{Probability calculation by random forest}
Although there are many other candidate methods to calculate the probability $\Pr(z=i|\theta_{k}^{\ast})$, we use random forest (\cite{Breiman2001}).
Random forest is one of the ensemble learning methods of decision trees.
Random forest is attractive because they are an ensemble of decision trees that can partition the parameter space effectively.
This characteristic makes them particularly well-suited for identifying the reconstructable area in which the sub-posterior distributions overlap.
In Addition, its non-parametric nature allows it to adapt to the complex structure of posterior distributions, which are often multimodal.
\cite{Wang2015} used random partition trees, similar to random forest, and demonstrated the efficiencies. 
As a result, random forest is suited for calculating the probability $\Pr(z=i|\theta_{k}^{\ast})$.

\subsection{Posterior approximation and KL divergence-based criterion}
We have already explained that machine learning classifiers, especially random forest, can calculate the probability of belonging to the reconstructable area.
However, in general, training a classifier model requires hyperparameter tuning.
We use random forest for the proposed parallelizing MCMC, not for simple classification tasks; thus, an optimization criterion is required.
We introduce a method for extracting MCMC draws from the reconstructable area using the proposed probability calculation and approximating the full-posterior distribution.
For efficient extraction, we also propose an optimization criterion for training a classifier: an upper bound of the Kullback-Leibler (KL) divergence between the full-posterior and approximated posterior distributions using only the probabilities from a classifier $\Pr(z_{k}=i|\theta_{k}^{\ast})$ and sub-posterior density values.
In almost all MCMC methods, such as the Metropolis-Hastings and Hamiltonian Monte Carlo methods (\cite{BDA}), sub-posterior density values are already obtained after each MCMC is completed.
Consequently, our proposed parallelizing MCMC method reduces the effort required for hyperparameter tuning.

Here, we explain how to approximate the full-posterior distribution.
We assume $Q$ and $P$ are categorical distributions with a category size $m$, respectively. 
For any category $j \in {1,\ldots,m}$, we denote by $Q(z=j)$ and $P(z=j)$ their respective probabilities evaluated in category $j$.
From the information geometry and theory of large deviation (\cite{Amari_undated-bp}), the probability that $P$ is equal to $Q$ is given by
\begin{align*}
    \Pr{(P = Q)} = \exp{(-\text{KL}(Q||P))},
\end{align*}
where $\text{KL}(Q||P))$ is the Kullback–Leibler divergence between $Q$ and $P$.
Now, for $j = 1,\ldots,m$, we define $Q(z=j) = 1/m$, which is the baseline probability distribution for searching the reconstructable area discussed in Subsection \ref{Strategy of combining posterior draws}.
We also define the categorical distribution $P_{i}^{(t)}$ with category size $m$ and probability $P_{i}^{(t)}(z = j) = \Pr{\left(z=j|\theta_{i}^{(t)}\right)}$.
Then, the probability above given a posterior draw $\theta_{i}^{(t)}$, $t=1,\ldots,T$, is
\begin{align*}
    \Pr{\left(P_{i}^{(t)}=Q\right)} = \exp{\left(\log{m} + \frac{1}{m}\sum_{j=1}^{m}\log{\Pr{\left(z=j|\theta_{i}^{(t)}\right)}}\right)}.
\end{align*}
This probability can be interpreted as follows: the closer $\Pr{\left(P_{i}^{(t)}=Q\right)}$ is to 1, the more probable $\theta_{i}^{(t)}$ is a sample from the reconstructable area, and vice versa because $P_{i}^{(t)}(z=j)$ is close to $1/m$ in that case.
We assume that the higher the probability $\Pr{\left(P_{i}^{(t)}=Q\right)}$ is, the more frequently $\theta_{i}^{(t)}$ should be resampled.
By normalizing the probability for $i=1,\ldots,m$ and $t=1,\ldots,N$, we obtain a discrete distribution with the probability of $\theta_{i}^{(t)}$ occurring as follows:
\begin{align}
    f\left(\theta = \theta_{i}^{(t)}\right) = \frac{\Pr{\left(P_{i}^{(t)}=Q\right)}}{\sum_{i^{'}}\sum_{t^{'}}\Pr{\left(P_{i^{'}}^{(t^{'})}=Q\right)}} = \frac{\Pr{\left(P_{i}^{(t)}=Q\right)}}{C_{f}},
    \label{eq:approx_posterior}
\end{align}
where $C_{f} = \sum_{i^{'}}\sum_{t^{'}}\Pr{\left(P_{i^{'}}^{(t^{'})}=Q\right)}$ is the normalizing constant.
We assume that the distribution approximates the full-posterior distribution $\pi(\theta|X)$.

When we implement MCMC methods in parallel, we generally calculate and obtain the value of the density $\pi(\theta_{i}^{(t)}|X_{i})$, $i=1,\ldots,m$ and $t=1,\ldots, T$, thus we can use them for other calculations after implementing each MCMC.
Let $f_{\pi}$ be the discrete full-posterior distribution:
\begin{align*}
    f_{\pi}\left(\theta_{i}^{(t)}\right) = \frac{\pi\left(\theta_{i}^{(t)} | X\right)}{C_{\pi}},
\end{align*}
where $C_{\pi} = \sum_{i^{'}}\sum_{t^{'}} \pi\left(\theta_{i^{'}}^{(t^{'})} | X\right)$.
Let also $C_{\pi_{\text{sub}}} = \sum_{i^{'}}\sum_{t^{'}} \pi\left(\theta_{i^{'}}^{(t^{'})} | X_{i^{'}}\right)$.
Using these values, we can evaluate the Kullback-Leibler divergence between $f_{\pi}(\theta)$ and $f(\theta)$.






\begin{thm} \label{upperbound of KL}
{\rm 
A constant $H > 0$, which is unrelated to classifier model $\Pr(z|\theta)$, exists and $\text{KL}(f_{\pi}||f)$ is bounded as
}
{\rm
\begin{align*}
    \frac{1}{H}{\text{KL}}(f_{\pi}||f) \le \log{\frac{C_{f}}{m}} - \frac{1}{m} \sum_{i}\sum_{t}\frac{\pi\left(\theta_{i}^{(t)}|X_{i}\right)}{C_{\pi_{\text{sub}}}} \sum_{j=1}^{m}\log{\Pr{\left(z=j|\theta_{i}^{(t)}\right)}}.
\end{align*}
}
\end{thm}
\begin{proof}
    See Appendix \ref{Proof of Theorem}.
\end{proof}

\begin{remark}
{\rm
The constant $H$ is not related to the model $\Pr(z|\theta_{k}^{\ast})$, therefore in the comparison of different approximated posteriors $f$ and $f^{'}$, we can just use the right-hand side of Theorem \ref{upperbound of KL}.
Hereafter, we refer to the right-hand side of Theorem \ref{upperbound of KL}, that is, the upper bound of KL.
}
\end{remark}


As a summary of this section, we propose a parallelizing MCMC algorithm.
Algorithm 1 shows the details.

\begin{algorithm} \label{palalgorithm}
\caption{Parallelizing MCMC algorithm based on machine learning classifier with the upper bound of KL}
\begin{algorithmic}[1] 
\Procedure{MachineLearningParallelizingMCMC}{$X$, params} 
    \State Divide data $X$ into sub-data $X_{i}$, $i=1,2,\ldots,m$
    \For{$i = 1$ to $m$}
        \State $\theta_{i}$, $\pi(\theta|X_{i}) \leftarrow$ \Call{MCMC}{$X_{i}$}
        \State Store $i$, $\theta_{i}$, $\pi(\theta|X_{i})$
    \EndFor
    \State Initialize collection arrays $z = [~]$, $\Theta = [~]$, $p = [~]$
    \For{$i = 1$ to $m$}
        \State Append $i$ to $z$
        \State Append $\theta_{i}$ to $\Theta$
        \State Append $\pi(\theta|X_{i})$ to $p$
    \EndFor
    \State Initialize collection arrays $\Pr = [~]$, $\text{UKL} = [~]$
    \For{$j = 1$ to $|\text{params}|$}
        \State Append \Call{Classifier}{$\text{params}[j]$, $z$, $\Theta$} to $\Pr$
        \State Append \Call{UpperBoundKL}{$p$, $\Pr[j]$} to $\text{UKL}$
    \EndFor
    \State $k \leftarrow \arg\hspace{-0.5mm}\min \text{UKL}$
    \State Construct approximated posterior with (\ref{eq:approx_posterior}), $\Pr[k]$, and $\theta_{i}$'s.
\EndProcedure
\end{algorithmic}
\end{algorithm}


\section{Simulation} \label{simulation}
In this section, we evaluate the proposed parallelizing MCMC method using simulated datasets.
First, we implemented a multivariate normal distribution case in which the posterior distribution is multivariate normal.
Second, we simulated a mixed normal distribution case by using a multimodal posterior.
We used the $\{\text{ranger}\}$ package (\cite{ranger}) in R programming language to train the random forest.

\subsection{Multivariate normal distribution case}\label{multi normal experiment}
This subsection shows that the proposed upper bound of KL works by sampling from a multivariate normal posterior.
\subsubsection{Simulation settings}

In this case, we sample $N$ data points from $d$-dimensional multivariate normal described as
\begin{align*}
    X \sim N(\mu, \Sigma)
\end{align*}
where $\mu$ (unknown) and $\Sigma$ (known) represent the mean vector and covariance matrix of the distribution, respectively.
This dataset is then divided into $m$ subsets to implement the proposed parallelizing MCMC method.
We set the prior distributions of the parameters to be uniform:
\begin{align*}
    \pi(\mu) \propto \text{constant}.
\end{align*}

After data partitioning, we compute the full-posterior distribution using the undivided dataset and $m$ sub-posterior distributions using the $m$ divided subsets.
Given $m$ data partition $X = \{X_{1}, X_{2}, \ldots, X_{m}\}$ and the prior distributions, the full-posterior distribution is derived as follows:
\begin{align*}
    \pi(\mu | X) &\sim N\left(\mu_{\text{post}}, \Sigma_{\text{post}}\right),
\end{align*}
where $\mu_{\text{post}} = \sum_{x \in X} x / N$ and $\Sigma_{\text{post}} = \Sigma / N$.
The sub-posterior distributions for \(i = 1, 2, \ldots, m\) are calculated as:
\begin{align*}
    \pi(\mu | X_i) &\sim N\left(\mu_{\text{post}, i}, \Sigma_{\text{post}, i} \right),
\end{align*}
where $N_{i}$ is the number of data in $X_{i}$, $\mu_{\text{post}, i} = \sum_{x_{i} \in X_{i}} x_{i} / N_{i}$, and $\Sigma_{\text{post}, i} = \Sigma / N_{i}$.

After implementing the proposed parallelizing MCMC method, we calculate the distance between the full-posterior and the approximated posterior using the Kullback-Leibler divergence (KL) as true KL and the proposed upper bound of KL.
Let $\theta_{\text{full}}^{(t)}$ and  $\theta_{\text{approx}}^{(t)}$ be $t$-th full-posterior draw and approximated posterior draw, respectively, by the proposed parallelizing MCMC method.
In addition, let $\hat{\mu}_{\text{full}}$ and $\hat{\mu}_{\text{approx}}$ be the posterior means, and $\hat{\Sigma}_{\text{full}}$ and $\hat{\Sigma}_{\text{approx}}$ be the posterior covariance matrices.
The sample-based KL between the normal distributions can be calculated as follows:
\begin{align}
\mbox{KL}(\pi||f)&=\frac{1}{2}\left(\mbox{tr}(\hat{\Sigma}_{\text{approx}}^{-1}\hat{\Sigma}_{\text{full}})+(\hat{\mu}_{\text{approx}}-\hat{\mu}_{\text{full}})^{\top}\hat{\Sigma}_{\text{approx}}^{-1}(\hat{\mu}_{\mbox{\scriptsize approx}}-\hat{\mu}_{\mbox{\scriptsize full}}) \right. \nonumber \\
& \qquad \left. -d + \log\frac{|\hat{\Sigma}_{\text{approx}}|}{|\hat{\Sigma}_{\text{full}}|}\right),
\end{align}
where $\pi$ and $f$ are the full- and approximated posterior distributions, respectively.

To implement this experiment, we set $m=5$, $d = 10$, and $\mu = \boldsymbol{0}$ which is a $d$-dimensional zero vector.
We also set the $(g,l)$ element of $\Sigma$ to $0.8 \times |g-l|$.
To check the relationship between the true KL and the proposed upper bound of KL, we train random forest 50 times following the hyperparameters and their candidate values:
\begin{itemize}
    \item Friction: 0.8 to 0.999 in steps of 0.001. This parameter controls the subsampling fraction of the data used in each tree.
    \item Number of Trees: 10 to 100 in steps of one. This parameter specifies the total number of trees that grow in the forest.
    \item Minimum Node Size: 5 to 50 in steps of one. This parameter sets the minimum size of terminal nodes.
    \item \texttt{mtry}: 1 to the number of dimensions ($d=10$ in this case) in steps of one. This parameter determines the number of variables required to grow the tree.
    \item Replacement: TRUE or FALSE. This parameter specifies whether the data sampling should be performed with replacement.
    \item Weight Adjustment Method: TRUE or FALSE. This parameter indicates the method used to adjust the weights of observations. Weights based on the log sub-posterior density value are used.
\end{itemize}
We randomly selected the hyperparameters in each iteration from the sets above, trained the model, and assessed the correlation between the true KL and the proposed upper bound of KL.


When sampling from the approximated posterior, we implemented data augmentation to avoid obtaining duplicate draws by resampling from the original sub-posterior draws used in training.
Specifically, after training a random forest, we created augmented draws by multiplying the sub-posterior draws $\theta_{i}^{(t)}$ by random numbers generated from $\text{Uniform}(1/3, 3)$.
This augmentation introduces various candidate draws while maintaining the general characteristics of the approximated posterior distribution.
The trained classifier then evaluates these augmented draws to construct an approximated posterior distribution (\ref{eq:approx_posterior}).
If an augmented draw falls outside the reconstructable area, it receives a low weight and, thus, is rarely resampled.
Note that this augmentation does not affect the calculation of the upper bound of KL, because it is only used in the final resampling stage.

\subsubsection{Result}

Figure \ref{fig:multivariate_normal_result} shows how the proposed upper bound of KL works.
The proposed upper bound of KL correlates well with the true KL between the full-posterior distribution and the proposed approximated posterior distribution.
Thus, we can use the proposed KL criterion to select machine learning classifiers to calculate the probability of belonging to the reconstructable area.
Figure \ref{fig:multivariate_scatter_matrix} shows that the proposed method effectively captures the reconstructable area.
However, the tails of the approximated posterior distribution spread because of the data augmentation to avoid duplicate resampling.

\begin{figure}[tbp]
 \begin{center}
  \includegraphics[width=120mm]{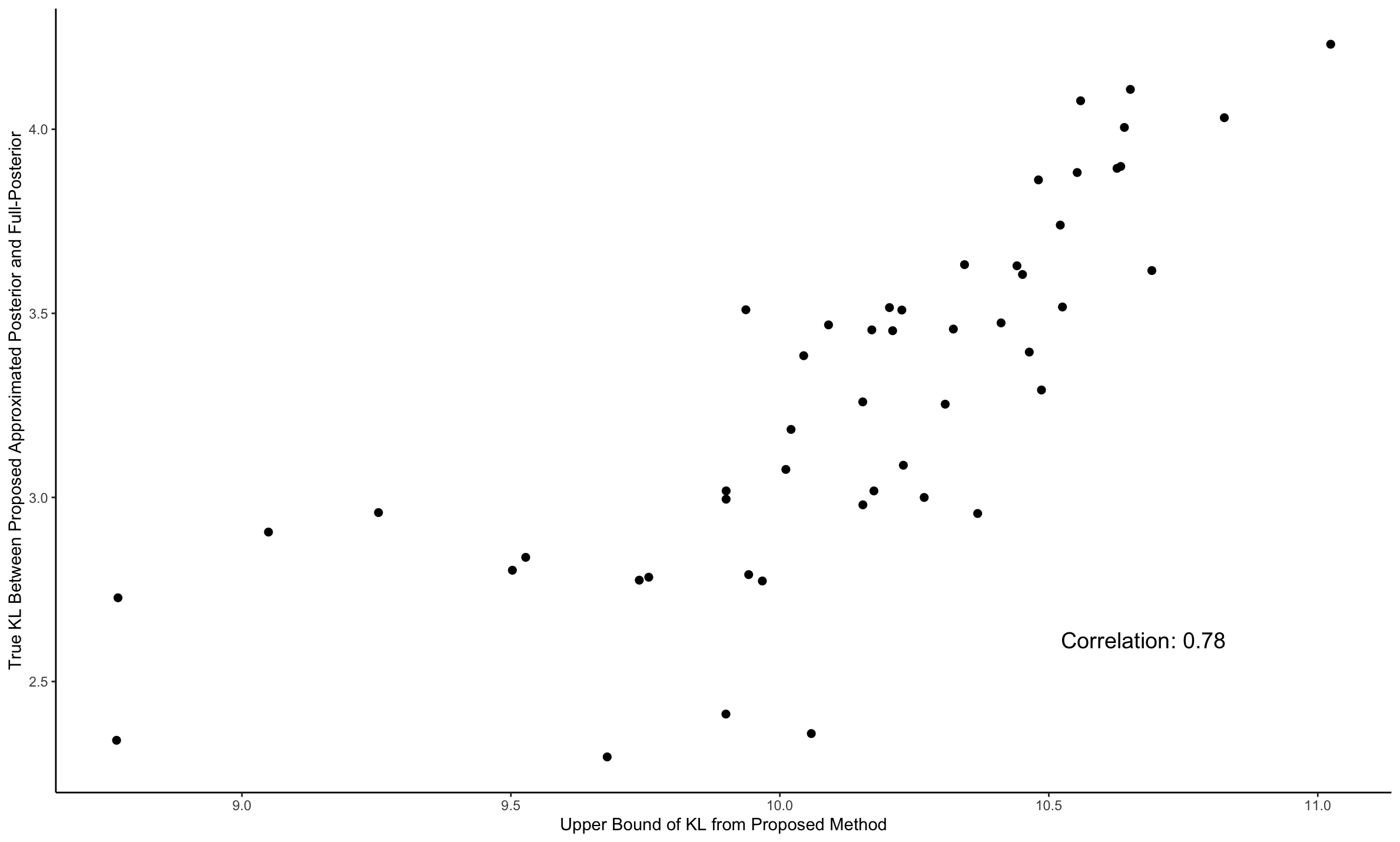}
 \end{center}
 \caption{Between true posterior and approximated posterior, the scatter plot of the proposed upper bound of KL and true KL.}
 \label{fig:multivariate_normal_result}
\end{figure}

\begin{figure}[tbp]
 \begin{center}
  \includegraphics[width=150mm]{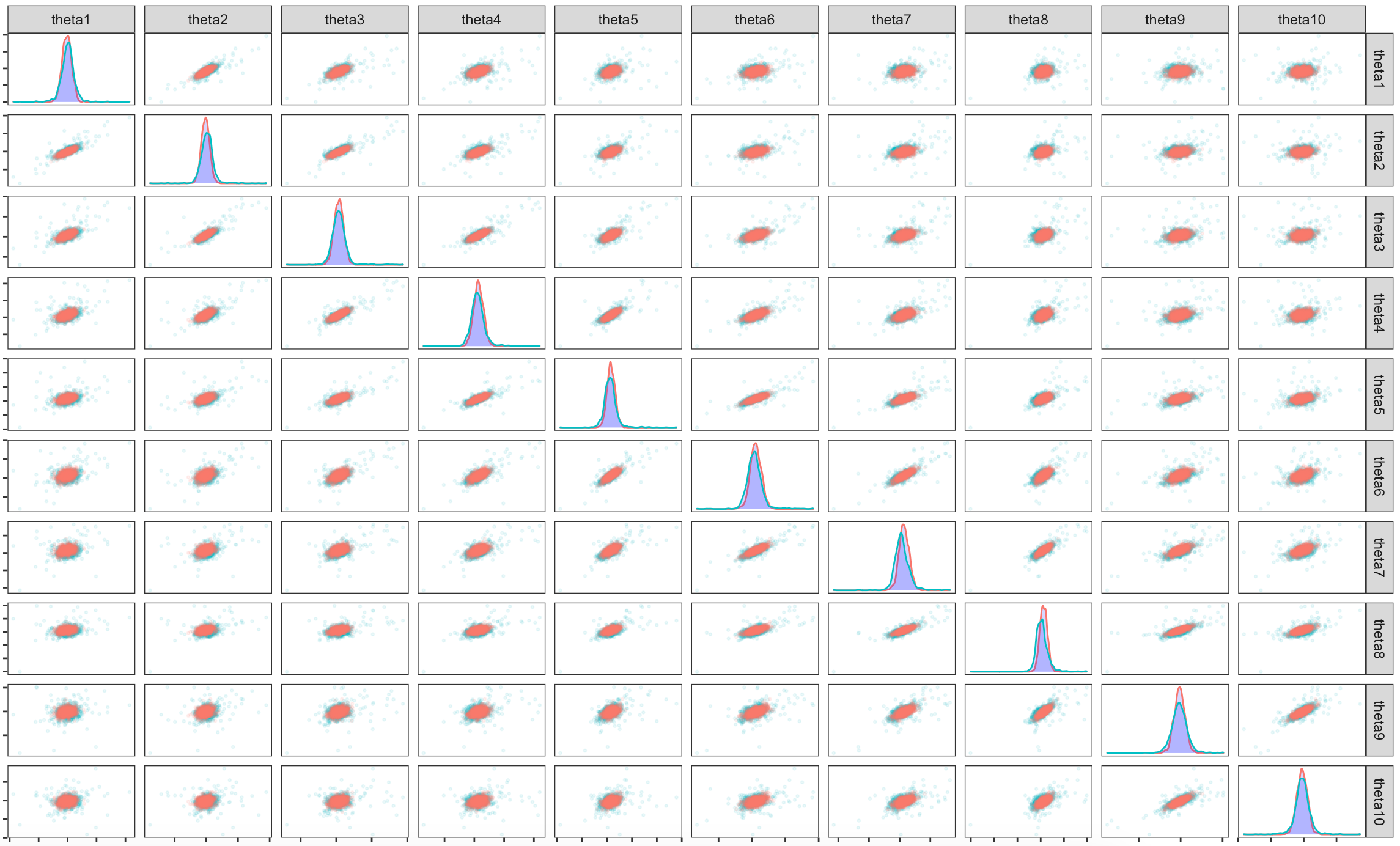}
 \end{center}
 \caption{Scatter matrix plot. The red-colored points and density show the full-posterior distribution, and the light green-colored points and density show the approximated posterior distribution calculated with the best-performed parameter settings in Figure \ref{fig:multivariate_normal_result}.}
 \label{fig:multivariate_scatter_matrix}
\end{figure}

Table \ref{tab:comparison} compares some parallelizing MCMC methods.
\cite{Scott2013}'s and \cite{Dunson2014}'s methods performed well because their methods depend on posterior normality, and the posterior distribution in this case is normally distributed.
\cite{Neiswanger2014}'s method suffered from approximating the full-posterior distribution through kernel density estimation on the term KL and combining time.
However, the proposed method in this study performed relatively well in terms of KL.
Although the combining time in our method is also relatively long, the random forest can efficiently handle large volumes of and high-dimensional data and remain a strong candidate for practice \citep{Qi2012, ranger}.

\begin{table}[h]
\centering
\begin{tabular}{lcc}
\hline \hline
\textbf{Method} & \textbf{True KL} & \textbf{Combining time (sec)} \\ \hline
\cite{Scott2013} & 0.189 & 0.067 \\ 
\cite{Neiswanger2014} & 42.758 & 65.599 \\ 
\cite{Dunson2014} & 0.012 & 1.949 \\ 
Proposed & 2.727 & 55.107 \\ \hline \hline
\end{tabular}
\caption{Comparison of Methods. In the proposed method, we used the hyperparameter setting, the smallest upper bound of KL, to train the random forest.}
\label{tab:comparison}
\end{table}

\subsection{Mixture normal case}\label{mixture_normal_case_study}
This subsection shows that the proposed parallelizing MCMC method works by sampling from a mixture of normal posterior distributions with multimodes.

\subsubsection{Simulation settings}
We adapted our methods and the preceding methods to a mixture of normal distribution cases in which the posterior is multimodal. 
Assume that we get $10,000$ data from: 
\begin{eqnarray}
X \sim \frac{1}{4}{N}(-3,1)+\frac{1}{2}{N}(0,1) + \frac{1}{4}{N}(3,1).
\label{eq:multimodalmodel}
\end{eqnarray}
Now, let us call the components on the right-hand side of (\ref{eq:multimodalmodel}) the 1st, 2nd, and 3rd components from left to right.
Suppose we also know that $X$ is sampled from which component in equation (\ref{eq:multimodalmodel}).
Let $L \in \{1, 2, 3\}$ be the label from which component $X$ is sampled.
We obtain $10,000$ samples of $D=(X, L)$.

With the obtained $D$'s, we equally divide $D$'s into five subsets, $D_{i}=(X_{i}, L_{i})$, $i=1,\ldots,5$, then we construct sub-posterior distributions with three modes as follows:
\begin{align*}
\pi(\theta|D_{i}) = \frac{1}{4}{N}\left(\theta\left|\bar{X}_{i1}, \frac{1}{N_{ i1}}\right.\right) + \frac{1}{2}{N}\left(\theta\left|\bar{X}_{i2}, \frac{1}{N_{i2}}\right.\right) + \frac{1}{4}{N}\left(\theta\left|\bar{X}_{i3}, \frac{1}{N_{i3}}\right.\right),
\end{align*}
where $\bar{X}_{ic}$, $c=1,2,3$, is the mean of $c$-th component of $X_{i}$ and $N_{ic}$ is the number of data points of $c$-th component in $i$-th data partition.
Thus, we assume that IPE is held, and the full-posterior distribution can be written as
\begin{align*}
\pi(\theta|D)\propto\prod_{i=1}^{5}\pi(\theta|D_{i}).
\end{align*}

When implementing the proposed parallelizing MCMC method, we use random forest to calculate the probability of belonging to the reconstructable area and determine the hyperparameters through random search in the same way as in subsection \ref{multi normal experiment}, to select the model with the minimum proposed upper bound of KL in the random search.

\subsubsection*{Result}
Measuring the distance between the full-posterior and approximated posterior distributions in a multimodal case is difficult.
Hence, we only provide their density traces in Figures \ref{fig:mixtureappliedposterior_comparison} and \ref{fig:mixtureappliedposterior}.

Figure \ref{fig:mixtureappliedposterior_comparison} compares some parallelizing MCMC methods.
As you can imagine in the provision, \cite{Scott2013}'s method merged several modes by taking a weighted average.
Both \cite{Neiswanger2014}'s and \cite{Dunson2014}'s methods over-represented only the densest central part, and this was especially the case for \cite{Dunson2014}'s method.
For visibility, we extract only the density traces of the full-posterior and approximated posterior distributions using the proposed method as shown in Figure \ref{fig:mixtureappliedposterior}.
This figure shows that the proposed method adequately captured the multimodes of the distribution.

\begin{figure}[tbp]
 \begin{center}
  \includegraphics[width=140mm]{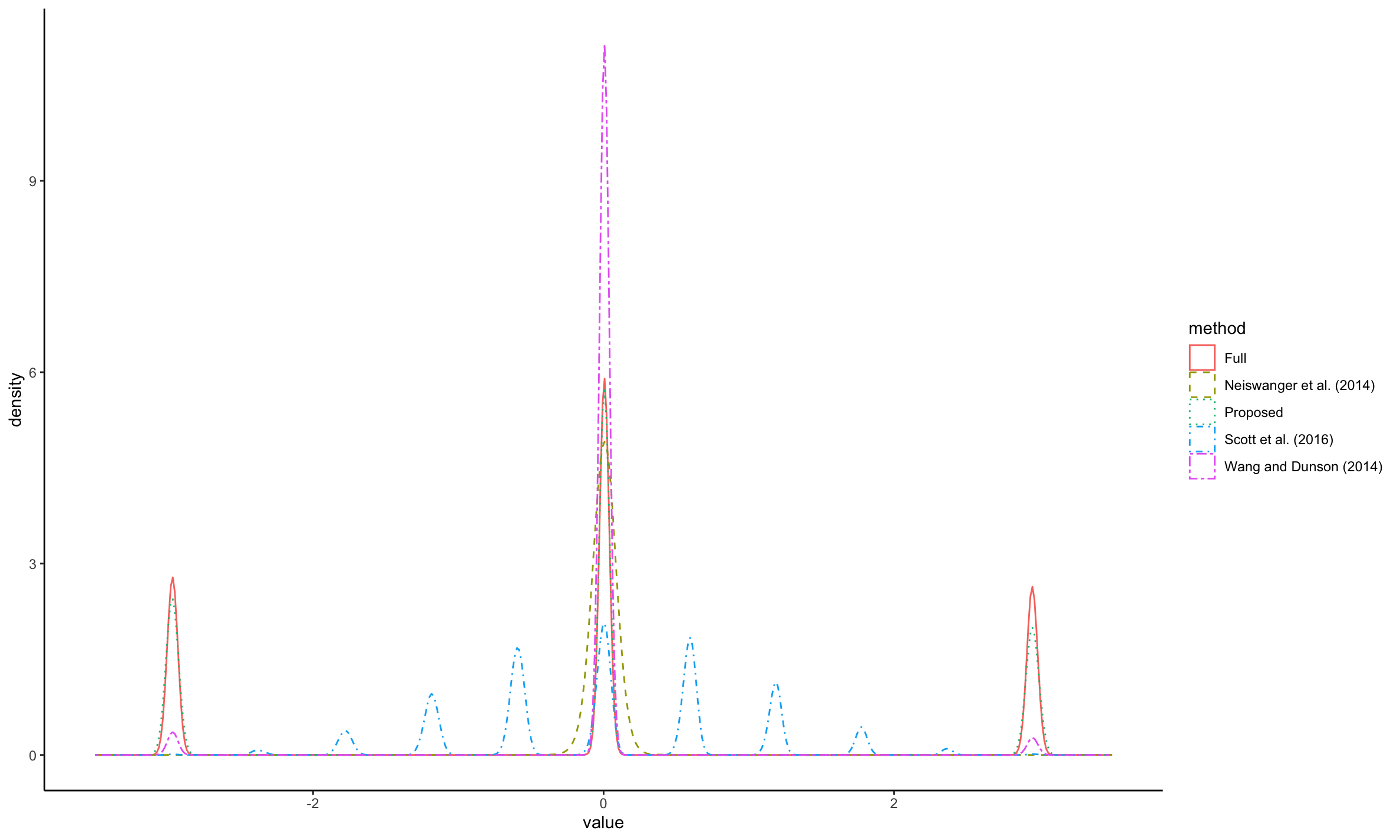}
 \end{center}
 \caption{The comparison of the density traces of the full-posterior and approximated posteriors by parallelizing MCMC methods.}
 \label{fig:mixtureappliedposterior_comparison}
\end{figure}

\begin{figure}[tbp]
 \begin{center}
  \includegraphics[width=140mm]{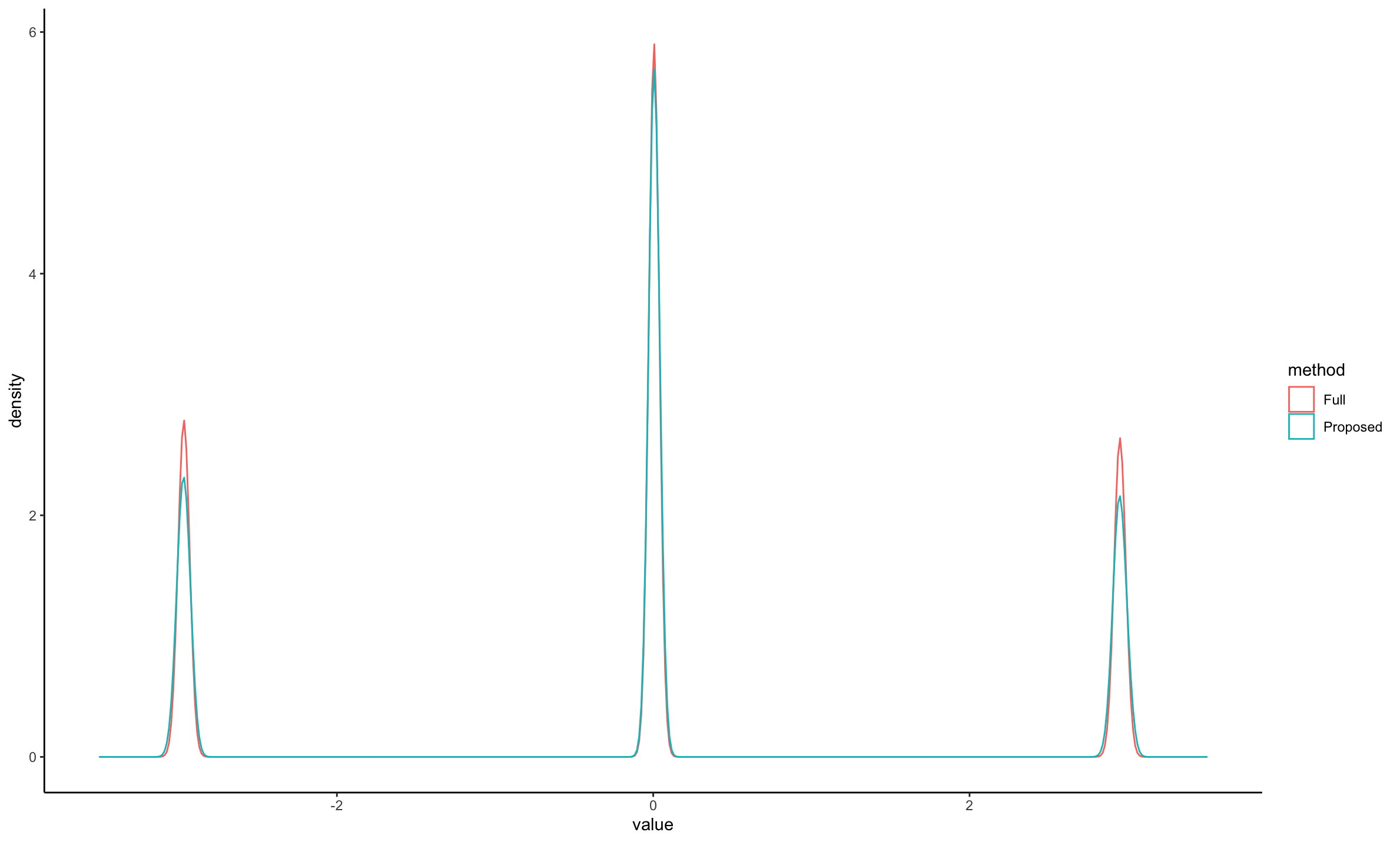}
 \end{center}
 \caption{The density traces of the full-posterior and approximated posterior by the proposed parallelizing MCMC method.}
 \label{fig:mixtureappliedposterior}
\end{figure}

\section{Conclusion} \label{conclusion}

This study addressed the computational challenges in parallelizing MCMC by focusing on how to effectively combine sub-posterior draws.
We introduced the key concept to combining these draws, the reconstructable area where sub-posterior distributions overlap.
We developed a parallelizing MCMC method using machine learning classifiers, particularly random forest, to extract posterior draws from the reconstructable area and proposed an algorithm for parallelizing MCMC.
We also developed the upper bound of KL divergence as a selection criterion for classifiers.
This criterion can be calculated using only sub-posterior MCMC information, making it particularly useful for tuning hyperparameters of classifiers.
Simulation studies demonstrated that our proposed upper bound of KL correlates well with the true KL divergence between full- and approximated posterior distributions, and confirmed the effectiveness of our method, especially for a posterior distribution with multiple modes.


We explain the two drawbacks of the proposed method.
First, efficient data augmentation methods are required to effectively sample posterior draws using the proposed method. 
Simply resampling the original sub-posterior draws used in the training results in many duplicate draws. 
Although we addressed this issue through data augmentation by applying multiplicative jittering to create new candidate draws, this approach has limitations. 
For machine learning classifiers, such data augmentation can lead to partial extrapolation, particularly affecting the tails of the approximated posterior distribution. 
By contrast, \cite{Dunson2014}'s method avoids this issue by obtaining new approximate posterior draws through Gibbs sampling. 
Although our method demonstrates the ability to reconstruct the full-posterior distribution and effectively partition the parameter space using machine learning classifiers, future research should focus on developing a theoretically consistent data augmentation method.

Second, in general, studies on parallelizing MCMC methods would need to overcome the curse of dimensionality. 
The curse of dimensionality is a frequent problem in many statistical and machine learning areas, and almost all parallelizing MCMC methods are no exception.
As the dimensionality of the model increases, the posterior draws in the parameter space would become sparsely distributed (\cite{DAMOUR2021}). 
This property makes extracting the reconstructable areas for high-dimensional models difficult and requires a reasonable number of iterations for each MCMC. 
In such a case, dividing the data and performing MCMC weakens the usefulness of parallelizing MCMC to reduce overall computational costs. 

\appendix

\section{Proof of Theorem \ref{upperbound of KL}} \label{Proof of Theorem}
\begin{proof}
Let $C_{\pi}^{'} = \sum_{i^{'}}\sum_{t^{'}} \prod_{j}\pi\left(\theta_{i^{'}}^{(t^{'})}|X_{j}\right)$, and $C_{\pi}^{'} = C_{\pi}$ because of IPE.
We can obtain a constant $H$ such that 
\begin{align*}
    \max_{i,t}\left\{\frac{\prod_{j}\pi(\theta_{i}^{(t)}|X_{j})}{C_{\pi}^{'}}\right\} \le H \times \min_{i,t}\left\{\frac{\pi(\theta_{i}^{(t)}|X_{i})}{C_{\pi_{\text{sub}}}}\right\},
\end{align*}
because the number related to $i$ and $t$ is at most finite.
Hence, note that $\log{f\left(\theta_{i}^{(t)}\right)} < 0$; then,
\begin{align*}
    & \sum_{i}\sum_{t} \left\{\frac{\pi\left(\theta_{i}^{(t)}|X\right)}{C_{\pi}} - \frac{\pi\left(\theta_{i}^{(t)}|X_{i}\right)}{C_{\pi_{\text{sub}}}}\right\} \log{f\left(\theta_{i}^{(t)}\right)} \\
    &= \sum_{i}\sum_{t} \left\{\frac{\prod_{j}\pi\left(\theta_{i}^{(t)}|X_{j}\right)}{C_{\pi}^{'}} - \frac{\pi\left(\theta_{i}^{(t)}|X_{i}\right)}{C_{\pi_{\text{sub}}}}\right\} \log{f\left(\theta_{i}^{(t)}\right)} \\
    &\ge (H - 1) \sum_{i}\sum_{t} \frac{\pi\left(\theta_{i}^{(t)}|X_{i}\right)}{C_{\pi_{\text{sub}}}} \log{f\left(\theta_{i}^{(t)}\right)}.
\end{align*}
Therefore,
\begin{align}
    H \sum_{i}\sum_{t} \frac{\pi\left(\theta_{i}^{(t)}|X_{i}\right)}{C_{\pi_{\text{sub}}}} \log{f\left(\theta_{i}^{(t)}\right) \le \sum_{i}\sum_{t} \frac{\pi\left(\theta_{i}^{(t)}|X\right)}{C_{\pi}} \log{f\left(\theta_{i}^{(t)}\right)}}. \label{eq:lemma}
\end{align}
By using (\ref{eq:lemma}), we obtain the following inequality:
\begin{align*}
    \text{KL}(f_{\pi} || f) &= \sum_{i}\sum_{t}\frac{\pi\left(\theta_{i}^{(t)}|X\right)}{C_{\pi}} \log{\frac{\pi\left(\theta_{i}^{(t)}|X\right)}{C_{\pi}}} - \sum_{i}\sum_{t}\frac{\pi\left(\theta_{i}^{(t)}|X\right)}{C_{\pi}} \log{f\left(\theta_{i}^{(t)}\right)}\\
    &\leq -\sum_{i}\sum_{t}\frac{\pi\left(\theta_{i}^{(t)}|X\right)}{C_{\pi}} \log{f\left(\theta_{i}^{(t)}\right)} \\
    &\le - H \sum_{i}\sum_{t}\frac{\pi\left(\theta_{i}^{(t)}|X_{i}\right)}{C_{\pi_{\text{sub}}}} \log{f\left(\theta_{i}^{(t)}\right)} \\
    &= H \left\{ \log{\frac{C_{f}}{m}} - \frac{1}{m} \sum_{i}\sum_{t}\frac{\pi\left(\theta_{i}^{(t)}|X_{i}\right)}{C_{\pi_{\text{sub}}}} \sum_{j=1}^{m}\log{\Pr{\left(z=j|\theta_{i}^{(t)}\right)}}\right\}.
\end{align*}
\end{proof}



\section*{Acknowledgement}
The author would like to thank Professor Yuichiro Kanazawa for his valuable comments.

\bibliography{ref}

\end{document}